\documentclass[12pt,draftclsnofoot,onecolumn]{IEEEtran}
\usepackage{amsfonts}
\usepackage{times}
\usepackage{latexsym}
\usepackage{amssymb}
\usepackage{amsmath}
\usepackage{cite}
\usepackage{verbatim}

\newcommand{\bydef}{\triangleq}

\def\SNR{{\textsf{SNR}}}

\def\bydef{:=}

\def\bb0{{\mathbb{0}}}

\def\bydef{:=}

\def\bb{{\mathbf{b}}}

\def\bee{{\mathbf{e}}}

\def\bg{{\mathbf{g}}}
\def\bh{{\mathbf{h}}}

\def\bw{{\mathbf{w}}}

\def\b0{{\mathbf{0}}}

\def\bbC{{\mathbb{C}}}

\def\bbE{{\mathbb{E}}}

\def\bbN{{\mathbb{N}}}

\def\bbP{{\mathbb{P}}}

\def\bbR{{\mathbb{R}}}

\def\bydef{:=}

\def\sf0{{\mathsf{0}}}

\usepackage{graphicx}
\usepackage{amssymb}
\usepackage{amsfonts}
\usepackage{amsmath}
\usepackage{latexsym}
\begin{document}

\newtheorem{thm}{Theorem}
\newtheorem{lemma}{Lemma}
\newtheorem{rem}{Remark}
\newtheorem{exm}{Example}
\newtheorem{prop}{Proposition}
\newtheorem{defn}{Definition}
\newtheorem{cor}{Corollary}
\def\proof{\noindent\hspace{0em}{\itshape Proof: }}
\def\endproof{\hspace*{\fill}~\QED\par\endtrivlist\unskip}
\def\bh{{\mathbf{h}}}
\def\SNR{{\mathsf{SNR}}}
\title{ Two-Way Transmission Capacity of Wireless Ad-hoc Networks}
\author{Rahul~Vaze, Kien T. Truong, Steven Weber and Robert W. Heath Jr.
\thanks{R. Vaze is with the School of Technology and Computer Science, Tata Institute of Fundamental Research, Homi Bhabha Rd, Mumbai, 400005, India,~ (email:~vaze@tcs.tifr.res.in).}
\thanks{K.~T.~Truong,~and~R.~W.~Heath,~Jr.~are~with~Department~of~Electrical~and~Computer~Engineering,
~the~University~of~Texas~at~Austin,~Austin,~TX~78712,~USA,\
(email:~kientruong@mail.utexas.edu, rheath@ece.utexas.edu).}
\thanks{S.~Weber~is~with~Department~of~Electrical~
and~Computer~Engineering,~Drexel~University,~Philadelphia,~PA~19104,~USA,~(email:~sweber@ece.drexel.edu).}
\thanks{This work is funded by DARPA through IT-MANET grant no. W911NF-07-1-0028.}
}

\date{}
\maketitle
%\markboth{To be submitted to IEEE Trans. on  Information Theory Draft Version 5}
\noindent
\begin{abstract}
The transmission capacity of an ad-hoc network is the maximum density of active transmitters per unit area, given an outage constraint at each receiver for a fixed rate of transmission. 
Most prior work on finding the transmission capacity of ad-hoc 
networks has focused only on one-way communication where a source 
communicates with a destination and no data is sent from the destination to the 
source. In practice, however, two-way or 
bidirectional data transmission is required 
to support control functions like packet acknowledgements and channel 
feedback.  This paper extends the concept of transmission capacity 
to two-way wireless  ad-hoc networks by incorporating the concept of a 
two-way outage with different rate requirements in both directions.
Tight upper and lower bounds on  the two-way transmission capacity are derived for frequency division duplexing. The derived bounds are used to derive the optimal solution for bidirectional
bandwidth allocation that maximizes the two-way transmission capacity, which is shown to perform better than allocating bandwidth proportional to the desired rate in both directions. Using the proposed two-way transmission capacity framework, a lower bound on the two-way transmission capacity with transmit beamforming using limited feedback is derived as a function of bandwidth, and bits allocated for feedback. 
\end{abstract}

\section{Introduction}
The transmission capacity of an ad-hoc wireless network is the maximum allowable
spatial  density of transmitting nodes, satisfying a per transmitter receiver rate,
and outage probability constraint \cite{Weber2005,Weber2007, Weber2008, Baccelli2006}.
Essentially, the transmission capacity characterizes
the maximum number of transmissions per unit area that can be simultaneously
supported in an ad-hoc network under a quality of service constraint.
The transmission capacity framework allows the application of the rich tool set of stochastic geometry to derive closed-form bounds for the interference distribution in a spatial network when the locations of nodes form a Poisson point process (PPP)~\cite{Kingman93}.

In prior work, the transmission capacity has been used successfully to characterize the effect of various physical and medium access (MAC) layer techniques on the ad-hoc
network capacity, such as successive interference cancelation \cite{WeberAndrews07a},
multiple antennas \cite{Hunter2008, Jindal2008a, Huang2008, Vaze2009},
and guard-zone based scheduling \cite{HasanAndrews07}. 
Most of the prior work on finding the transmission capacity has been limited
to one-way communication (no data communication from the destination to the source), and precludes the possibility of two-way communication. In two-way (bidirectional) communication the destination also has data to send to its source,
e.g. channel state feedback~\cite{LoHeath08}, packet acknowledgement \cite{SinghKankipati04}, or route initiation and update requests~\cite{RoyerToh99}.

In this paper we define the two-way transmission capacity, and derive tight upper and lower bounds on it  when the transmitter location are distributed as a Poisson point process (PPP) distributed. The bounds are used to characterize the dependence of the two-way transmission capacity on the key  system parameters, e.g. bandwidth allocation in two directions given a data rate requirement. We
consider an ad-hoc network with two-way communication,
where each source destination pair has data to exchange in both directions.
We consider a general system where the data requirement in both directions can
be different, and a frequency division duplex (FDD) communication model,
where two separate frequency carriers are used for two directions, thereby forming a full-duplex link.

In a two-way communication model, where the transmitter
locations are modeled as a PPP, the interference received in
both directions is correlated, and hence the joint success probability in two directions
is not equal to the product of the success probabilities in
each direction. Therefore finding the exact expression for the
joint success probability is complicated.
To obtain meaningful insights on
the two-way transmission capacity, we  derive tight upper and
lower bounds on the two-way transmission capacity with
FDD, assuming that the channel coefficients on separate
frequencies are independent and all the channel coefficients
are Rayleigh distributed. The upper and lower bound only
differ by a constant, i.e. the bounds have identical
dependence on the parameters of interest (rate requirements,
and bandwidths allocated in each direction). Thus, the derived bounds establish the two-way
transmission capacity up to a constant.

The results of this paper in part have been presented in \cite{Kien2009, VazeISIT2010}. The differences between \cite{Kien2009, VazeISIT2010} and the present paper are as follows. For simplification of
analysis, \cite{Kien2009} assumed that the interference received in both
directions is independent. 
The independence assumption was removed in \cite{VazeISIT2010}, and upper and lower bounds on the two-way transmission capacity that derived which were shown to be tight. 
Compared to \cite{VazeISIT2010}, the present paper  
extends the two-way transmission capacity framework to quantify the loss in transmission capacity 
with practical limited feedback \cite{Love2008JSAC} in comparison to genie-aided feedback (channel coefficients are known exactly, and without any cost at the transmitter), when 
the transmitter is equipped with multiple antenna and uses  beamforming to transmit its signal to the receiver.  In addition to this, the present paper offers more clarity of exposition, complete proof of Theorem $2$, and added simulation results for more insights into the effects of two-way communication.

Using the derived bounds on the two-way transmission capacity, we find the optimal bandwidth allocation in two directions that
maximizes the transmission capacity. The optimal bandwidth allocation problem
is shown to be a convex program in a single variable which can be solved easily
by finding the value where the function derivative is zero.
Using the optimal bandwidth allocation solution, we show that an intuitive
strategy that allocates the bandwidth in proportion to the desired
rate in each direction is optimal only for symmetric traffic (same rate requirement in both
directions) and performs poorly for asymmetric traffic in comparison to the optimal strategy. Examples of asymmetric traffic are channel feedback, and ack/nack messages, where there is huge disparity between the data rates required in two directions.

There is extensive related work on resource allocation in wireless ad hoc networks, but almost all of it focused on one-way communication. For instance, prior work studied the spectrum sharing between two one-way spatial networks in~\cite{GrokopTse08},  between a spatial network
and a cellular uplink network in~\cite{HuangLau08}, and one-way spatial networks where the total bandwidth is optimally split into sub-bands to maximize the transmission
capacity \cite{Jindal2007WCOM}. Our bandwidth allocation, however, studies the bandwidth sharing between two directions within a single two-way spatial network.

As an application of the proposed  two-way transmission capacity framework, we evaluate the performance degradation with practical limited channel feedback  in comparison to genie aided channel feedback, when the transmitter has multiple antennas and uses beamforming for transmitting its signal to the receiver. 
We account for both the bandwidth used, and the bits required for feedback, to derive a lower  bound on the two-way transmission capacity with 
transmit beamforming using limited feedback. We show that with practical limited channel feedback, the two-way transmission capacity is  substantially reduction compared to the genie-aided case. The severe degradation results because with increasing the number of feedback bits, the transmission capacity increases sub-linearly  due to improvement in signal strength, however, decreases  exponentially because of the stringent requirement of feedback bits to be correctly decoded.

{\it Notation:}
The expectation of function $f(x)$ with respect to $x$ is denoted by
${\bbE}(f(x))$.
A circularly symmetric complex Gaussian random
variable $x$ with zero mean and variance $\sigma^2$ is denoted as $x
\sim {\cal CN}(0,\sigma^2)$. Let $S_1$ be a set and $S_2$ be a subset of $S_1$. Then $S_2 \backslash S_1$ denotes the set of elements of $S_1$ that do not belong to $S_2$. The integral $\int_{0}^{\infty}x^{k-1}e^{-x}dx$ is
denoted by $\Gamma(x)$. We use the symbol
$\bydef$ to define a variable.

%{\it Organization:} The rest of the paper is organized as follows.
%In Section \ref{sec:sys}, we describe the system model under consideration for
%the two-way ad-hoc network. In Section \ref{sec:twowaytc}, we derive tight upper and lower bound on
%the two-way transmission capacity.
%In Section \ref{sec:balloc} using the derived expression for the two-way transmission capacity, the problem of optimal bandwidth allocation is solved. A lower bound on the two-way transmission capacity with limited feedback beamforming is derived in Section \ref{sec:LimitedFeedback}. Simulation results are presented in Section \ref{sec::NumericalSimulationResults} followed by conclusions  in Section \ref{sec:Conclusions}. 

\section{System Model}
\label{sec:sys} Consider an ad-hoc network with two sets of nodes
${\cal T} \bydef \{Tx_n, \ n\in \bbN\}, $ and ${\cal R} = \{Rx_n, \ n\in \bbN\}$, where $Tx_n$ and $Rx_n$ want to exchange data between each other for each $n$. We assume that each
$Tx_n$ and $Rx_n$ have a single antenna.
We consider a slotted Aloha  random access protocol,
where at any given time, the pair $(Tx_n, \ Rx_n)$ transmits data to each other with
an access probability $P_a$ for each $n$, independently of all other nodes.
We assume that the distance between each $Tx_n$ and $Rx_n$ is $d$.
Let the location of $Tx_n$ be $T_n$, and $Rx_n$ be $R_n$.
The set $\Phi_T =\{T_n\}$ is modeled as a homogenous PPP on a two-dimensional plane with intensity $\lambda_0$,
similar to \cite{Weber2005, Weber2007, Huang2008}.
Since $R_n$ is at a fixed distance $d$ in a random direction from the $T_n$, the set $\Phi_R \bydef
\{R_n\}$ is also a homogenous PPP on a two-dimensional plane
with intensity $\lambda_0$, because of the random translation invariance property of PPP \cite{Daley2003}. Because of the assumed Aloha random access protocol, at any
given time, the active transmitter receiver location processes
$\Phi_T^a \bydef \{T_n | Tx_n \ \text{ is active} \}$, and
$\Phi_R^a \bydef \{R_n | Rx_n \ \text{ is active} \}$ are
homogenous PPPs  on a two-dimensional plane with intensity
$\lambda = P_a\lambda_0$.
We consider a frequency division duplex
system, where the total available bandwidth is $F_{total}$, out of which
$F_{TR}$ is dedicated for $Tx_{n} \rightarrow Rx_{n}$ communication to support a
rate demand $B_{TR}$ bits for all $n$,
and the rest $F_{RT} \bydef F_{total} -F_{TR}$ for the $Rx_{n} \rightarrow Tx_n$ communication
to support a rate demand of $B_{RT}$ bits for all $n$.

In a time slot when the pair $(Tx_0, \ Rx_0)$ is active,
the received signal at receiver $Rx_0$ is
\begin{eqnarray}
\label{rxsigforward}
y_0 = \sqrt{P_t}d^{-\alpha/2}h_0x_0 + \sum_{T_n \in \Phi_T^a \backslash \{T_0\}}\sqrt{P_t}d_{Tn}^{-\alpha/2}h_{0n}x_n + z_0,
\end{eqnarray}
and the received signal at receiver $Tx_0$ is
\begin{eqnarray}
\label{rxsigreverse}
w_0 = \sqrt{P_t}d^{-\alpha/2}g_0u_0 + \sum_{R_n \in \Phi_R^a \backslash \{R_0\}}\sqrt{P_t}d_{Rn}^{-\alpha/2}g_{0n}u_n + v_0,
\end{eqnarray}
where $P_t$ is the transmit power, $h_0$  is the channel between
$Tx_0$ and $Rx_0$, and and $g_0$ is the channel from $Rx_0$ and $Tx_0$,  $h_{0n}$ and $g_{0n}$
is the channel between $Tx_n$ and $Rx_0$, and $Rx_n$ and $Tx_0$, respectively,
$d_{Tn}$ and $d_{Rn}$ are the distances between $Tx_n$ and $Rx_0$, and
$Rx_n$ and $Tx_0$, respectively, $\alpha > 2$ is the path loss exponent, $x_n, u_n \in {\cal C}N(0,1)$ are signals transmitted from $Tx_n$ and $Rx_n$, respectively, and $z_0, v_0$ is the additive white Gaussian noise. The ad-hoc network is assumed to be interference limited \cite{Weber2005}, thus we drop the noise contribution  from the 
received signal. We assume that  $h_0, g_0$, $h_{0n}$,  and $g_{0n}$ are independent and identically distributed with ${\cal CN}(0,1)$ to model a Rayleigh fading channel.

With the received signal model (\ref{rxsigforward}) and (\ref{rxsigreverse}), the signal to interference ratio (SIR) for the transmission from $Tx_0 \rightarrow Rx_0$ and from $Rx_0 \rightarrow Tx_0$ are 
\[SIR_{TR} \bydef \frac{d^{-\alpha}|h_0|^2}{\sum_{T_n \in \Phi_T^a \backslash \{T_0\}}d_{Tn}^{-\alpha}|h_{0n}|^2}, \ SIR_{RT} \bydef \frac{d^{-\alpha}|g_0|^2}{\sum_{R_n \in \Phi_R^a \backslash \{R_0\}}d_{Rn}^{-\alpha}|g_{0n}|^2}.\]

Assuming interference as noise,  the mutual information \cite{Cover2004} for the
$Tx_0$ to $Rx_0$ communication using bandwidth $F_{TR}$, and for the
$Rx_0$ to $Tx_0$ communication using bandwidth $F_{total} -F_{TR}$ are
\[MI_{TR} \bydef F_{TR}\log\left(1+SIR_{TR}\right) \ \text{ bits/sec}, \ MI_{RT} \bydef (F_{RT})\log\left(1+SIR_{RT}\right) \ \text{bits/sec}.\]

Recall that the rate requirement for the $Tx_0 \rightarrow Rx_0$ transmission 
is $B_{TR}$ bits, and for the $Rx_0 \rightarrow Tx_0$ communication is
$B_{RT}$ bits. Thus, to account for the two-way or bidirectional nature of
communication, we define the success probability (complement of the outage probability $\epsilon$) as the probability that communication in both directions is successful
simultaneously, i.e.
\begin{equation}
\label{psucstwoway}
P_{success} = P\left(MI_{TR} > B_{TR}, \ MI_{RT} > B_{RT} \right).
\end{equation}
Let $\lambda$ be maximum density of nodes per 
unit area that can support rate $B_{TR}$ from $Tx_0 \rightarrow Rx_0$, and $B_{RT}$  bits from $Rx_0 \rightarrow Tx_0$ with 
success probability $P_{success} =1-\epsilon$, using bandwidth $F_{total}$. 
\begin{defn}
The {\it two-way} transmission capacity $C_{\epsilon}$ is defined as
\label{deftc} 
\[C_{\epsilon} \bydef (1-\epsilon) \lambda\left(\frac{B_{TR}+B_{RT}}{F_{Total}}\right) \text{bits/sec/Hz/m}^2.\]
\end{defn} 
The problem to solve is to find the $\lambda$ and consequently $C_{\epsilon}$ 
for a given rate $B_{TR}, \ B_{RT}$, outage probability $\epsilon$ and bandwidth $F_{total}$ .

%The mutual information between $\bx_0$ and $\by_0$ is defined as \cite{Cover2004}
%\begin{eqnarray*}
%I(\bx_0;\by_0)& =& \log_e\det\left(\bI +
%\left(d^{-\alpha/}\bH_0\bH_0^{*}\right)\left(\sum_{T_n \in
%\Phi/T_0}
%d_i^{-\alpha}\bH_{0n}\bH_{0n}^{*}\right)\right)\end{eqnarray*}

To compute the success probability we consider a typical
transmitter receiver pair $(Tx_0, Rx_0)$. Using the stationarity of the homogenous PPP and Slivnyak's Theorem  [19] (Page 121), it
follows that the statistics of the signal received at the typical receiver are identical to that
of any other receiver. Hence the outage probability is invariant with the choice of the receiver. 
Slivnyak's Theorem also states that the locations of the interferers for the
typical transmitter and receiver $(Tx_0,Rx_0)$, i.e. $\Phi^a_T \backslash \{T_0\}$
and $\Phi^a_R \backslash \{R_0\}$ are also homogenous PPPs, each with intensity $\lambda$.

\section{Computing the Two-Way Transmission Capacity}
\label{sec:twowaytc}
In this section we derive an upper and lower bound on the two-way transmission capacity. 
To derive a lower bound we use the Fortuin, Kastelyn, Ginibre (FKG) inequality \cite{Grimmett1980}, while for deriving an upper bound we make use of the Cauchy-Schwartz inequality. 
Before stating the FKG inequality, we need the following definitions.
\begin{defn} A random variable $X$ defined on a probability space $(\Omega, {\cal F}, \bbP)$ is called increasing if $X(\omega) \le X(\omega')$ 
whenever $\omega \le \omega'$, for some partial ordering on $\omega, \ \omega' \in \Omega$. $X$ is called decreasing if $-X$ is increasing. 
\end{defn}
\begin{exm}\label{exmdec} $SIR_{TR}$ and $SIR_{RT}$ are decreasing random variables. 

For the PPP under consideration, let $\omega = (a_1, a_2, \ldots, )$ where for $n \in \bbN$,
\[a_n = \left\{\begin{array}{cc} 1 & \text{if $Tx_n$ is active,} \\
 0 & \text{otherwise.}\end{array}\right.\] Then $\omega' \ge \omega$, if $a_n'\ge a_n, \ \forall \ n$, i.e. configuration $\omega'$ contains at least those interferers which are present in configuration $\omega$. Recall our definition of $SIR_{TR} = \frac{d^{-\alpha}|h_0|^2}{\sum_{T_n \in \Phi_T^a \backslash \{T_0\}}d_{Tn}^{-\alpha}|h_{0n}|^2}$. Clearly, if there are more interferers present, $SIR_{TR}$ decreases, i.e. considering $SIR_{TR}$ as a random variable $SIR_{TR}(\omega) \ge SIR_{TR}(\omega')$, if $\omega \le \omega'$. Thus $SIR_{TR}$ is a decreasing random variable and so is $SIR_{RT}$. 
\end{exm}

\begin{defn}\label{exmdecevent} Let $A$ be an event in ${\cal F}$, and ${\cal I}_{A}$ be the indicator function of $A$. Then the event $A \in \ {\cal F}$ is called increasing if ${\cal I}_{A}(\omega) \le {\cal I}_{A}(\omega')$, whenever $\omega \le \omega'$. The event $A$ is called decreasing if its complement $A^{c}$ is increasing.

\end{defn}

\begin{exm}The success event $\{SIR >\beta\}$ is a decreasing event, since if $\omega' \in \{SIR >\beta\}$ and $\omega'\ge \omega$, then $\omega  \in \{SIR > \beta\}$.
\end{exm}

\begin{lemma}\label{lemfkg}(FKG Inequality \cite{Grimmett1980}) \newline
(a) If both $X$ and $Y$  are increasing or decreasing random variables with $\bbE\{X^2\} < \infty$, and $\bbE\{Y^2\} < \infty$, then 
$\bbE\{XY\} \ge \bbE\{X\}\bbE\{Y\}$.\newline
(b) If both $A, B \in {\cal F}$ are increasing or decreasing events then $P(AB) \ge P(A)P(B)$.
\end{lemma}

Now we are ready to derive bounds on the two-way transmission capacity. 
From (\ref{psucstwoway}), the success probability is
\[P_{success} = P\left(SIR_{TR} > 2^{\frac{B_{TR}}{F_{TR}}}-1, \
SIR_{RT} > 2^{\frac{B_{RT}}{F_{RT}}}-1\right).\]
Let $\beta_1 \bydef d^{\alpha}\left(2^{\frac{B_{TR}}{F_{TR}}}-1\right)$,
$\beta_2 \bydef d^{\alpha}\left(2^{\frac{B_{RT}}{F_{RT}}}-1\right)$, $I_{TR} \bydef \sum_{T_n \in \Phi_T^a \backslash \{T_0\}}d_{Tn}^{-\alpha}|h_{0n}|^2$, and 
$I_{RT} \bydef \sum_{R_n \in \Phi_R^a \backslash \{R_0\}}d_{Rn}^{-\alpha}
|g_{0n}|^2$.
Then, 
\begin{eqnarray} \nonumber
P_{success}& =&  P\left(\frac{|h_0|^2}{I_{TR}} > \beta_1, \
\frac{|g_0|^2}{I_{RT}} > \beta_2\right), \\ \nonumber
&\stackrel{(a)}{=}&\bbE\left\{e^{-\beta_1 I_{TR}}
e^{-\beta_2 I_{RT}}\right\}, \\\nonumber
& = & 
\bbE\left\{
e^{-
\beta_1 
\left(\sum_{T_n \in \Phi_T^a \backslash \{T_0\}}d_{Tn}^{-\alpha}|h_{0n}|^2\right)}
e^{-\beta_2 
\left(\sum_{R_n \in \Phi_R^a \backslash \{R_0\}}d_{Rn}^{-\alpha}
|g_{0n}|^2\right)}
\right\}, \\ \label{Psucintr}
&\stackrel{(b)}{=}& \bbE\left\{\prod_{T_n \in \Phi_T^a \backslash \{T_0\}}
\left(\frac{1}{1+\beta_1d_{Tn}^{-\alpha}}\right) 
\prod_{R_n \in \Phi_R^a \backslash \{R_0\}} 
\left(\frac{1}{1+\beta_2d_{Rn}^{-\alpha}}\right)
\right\},
\end{eqnarray}
where $(a)$ follows since $P\left(|h_0|^2 > x\right) = P\left(|g_0|^2 > x\right)=e^{-x}$, and $h_0$ and $g_0$ are independent, and $(b)$ follows by taking the expectation with respect to 
$h_{0n}$, and $g_{0n}$, and noting that 
$h_{0n}$, and $g_{0n}$ are independent and exponentially distributed.

The difficulty in evaluating the expectation with respect to $\{d_{Tn}\}$ and $\{d_{Rn}\}$ in the success probability (\ref{Psucintr}) lies in the fact that $d_{Tn}$ and $d_{Rn}$ are not independent. 
To visualize this, consider a network where there are
only two active pairs of nodes, $(Tx_0, Rx_0)$, and $(Tx_1, Rx_1)$
as depicted in Figure \ref{intfindp}. For the receiver $Rx_0$
receiving over bandwidth $F_{TR}$, the transmission from $Tx_1$ is
interference. As defined before, the distance between $Rx_0$ and
$Tx_1$ be $d_{T1}$. Thus, the interference power at $Rx_0$ is
$d_{T1}^{-\alpha}|h_{01}|^2$. Similarly, for $Tx_0$ receiving over
bandwidth $F_{RT}$, the transmission from $Rx_2$ is
interference. The distance between $Rx_1$ and $Tx_0$ be $d_{R1}$.
Thus, the interference power at $Rx_0$ is
$d_{R1}^{-\alpha}|g_{01}|^2$. For the case when $d$ is very small $d \rightarrow 0$, $d_{R1} \approx d_{T1}$, and thus distances $d_{R1}$ and $d_{T1}$ are not independent.
Moreover, explicitly computing the correlation between $d_{T_n}$ and $d_{R_n}$ is also a hard problem. Thus, to get a meaningful insight into the two-way transmission capacity we derive a lower and upper bound. 

{\bf \it Lower Bound:} Similar to Example \ref{exmdec},   $\prod_{T_n \in \Phi_T^a \backslash \{T_0\}}
\left(\frac{1}{1+\beta_1d_{Tn}^{-\alpha}}\right)$ and $\prod_{R_n \in \Phi_R^a \backslash \{R_0\}} 
\left(\frac{1}{1+\beta_1d_{Rn}^{-\alpha}}\right)$ are decreasing random variables, since each term in the product is less than $1$, and with the increasing the number of terms (number of interferers) in the product the total value of each expression decreases. Thus, using Lemma \ref{lemfkg}, from (\ref{Psucintr})
\begin{eqnarray} 
P_{success}& \ge & 
\bbE\left\{\prod_{T_n \in \Phi_T^a \backslash \{T_0\}}
\left(\frac{1}{1+\beta_1d_{Tn}^{-\alpha}}\right) \right\}
\bbE\left\{
\prod_{R_n \in \Phi_R^a \backslash \{R_0\}} 
\left(\frac{1}{1+\beta_2d_{Rn}^{-\alpha}}\right)
\right\}, \\ \nonumber
&\stackrel{(c)}{=}& e^{\left(-\lambda\int_{\bbR^2}1- 
\left(\frac{1}{1+\beta_1x^{-\alpha}}\right) \ dx\right)}e^{\left(-\lambda\int_{\bbR^2}1- 
\left(\frac{1}{1+\beta_2x^{-\alpha}}\right) \ dx\right)}, \\ \nonumber
&=& e^{\left(-2\pi\lambda\int_{0}^{\infty} 
\left(\frac{\beta_1x^{-\alpha+1}}{1+\beta_1x^{-\alpha}}\right) \ dx\right)}
e^{\left(-2\pi\lambda\int_{0}^{\infty} 
\left(\frac{\beta_2x^{-\alpha+1}}{1+\beta_2x^{-\alpha}}\right) \ dx\right)}, \\ 
\nonumber
&=& e^{-\lambda c_1  
\beta_1^{\frac{2}{\alpha}}}e^{-\lambda c_1  
\beta_2^{\frac{2}{\alpha}}}, \\ \label{lbound}
&=& e^{
-\lambda c_1  \left(
\beta_1^{\frac{2}{\alpha} }+
\beta_2^{\frac{2}{\alpha} } \right)
},
\end{eqnarray} 
where $(c)$ follows from the probability generating 
functional of the Poisson point process \cite[Example 4.2]{Stoyan1995}, and $c_1 = \frac{2\pi^2 Csc(\frac{2\pi}{\alpha})}{\alpha}$ is a constant, where $Csc$ is co-secant.

{\bf \it Upper Bound:} 
Using the Cauchy-Schwartz inequality, from (\ref{Psucintr})
\begin{eqnarray} \nonumber
P_{success}& \le & \left[
\bbE\left\{\prod_{T_n \in \Phi_T^a \backslash \{T_0\}}
\left(\frac{1}{1+\beta_1d_{Tn}^{-\alpha}}\right)^2 \right\}
\bbE\left\{
\prod_{R_n \in \Phi_R^a \backslash \{R_0\}} 
\left(\frac{1}{1+\beta_2d_{Rn}^{-\alpha}}\right)^2
\right\}\right]^{\frac{1}{2}},\\\nonumber
&\stackrel{(d)}{=}& 
\left[ 
e^{ 
-\lambda 
\left( \int_{\bbR^2}1- 
\left(\frac{1}{1+\beta_1x^{-\alpha}}\right)^2 dx\right)} 
e^{\left( 
-\lambda \int_{\bbR^2} 1 -\left(\frac{1}{1+\beta_2x^{-\alpha}} \right)^2 \ dx \right)} 
\right]^{\frac{1}{2}},\\ \nonumber
&=&\left[ 
e^{ 
-2\pi\lambda 
\left( \int_{\bbR^2}
\left(\frac{\beta_1^2x^{-2\alpha+1} + 2\beta_1x^{-\alpha+1}}{\left(1+\beta_1x^{-\alpha}\right)^2}\right) dx\right)} 
e^{ 
-2\pi\lambda 
\left( \int_{\bbR^2}
\left(\frac{\beta_2^2x^{-2\alpha+1} + 2\beta_2x^{-\alpha+1}}{\left(1+\beta_2x^{-\alpha}\right)^2}\right) dx\right)}\right]^{\frac{1}{2}},
 \\ \nonumber
&=& e^{-\lambda c_2  
\beta_1^{\frac{2}{\alpha}}}e^{-\lambda c_2  
\beta_2^{\frac{2}{\alpha}}},\\ \label{upbound}
&=& e^{
-\lambda c_2  \left(
\beta_1^{\frac{2}{\alpha} }+
\beta_2^{\frac{2}{\alpha} } \right)
},
\end{eqnarray}
where $(d)$ follows from the probability generating 
functional of the Poisson point process  \cite[Example 4.2]{Stoyan1995}, and $c_2 = \frac{\pi^2 Csc\left(\frac{2\pi}{\alpha}\right)\left(\alpha+2\right)}{\alpha^2}$ is a  constant, different from the constant $c_1$ of the lower bound.

\begin{thm}\label{thm:bounds} The two-way transmission capacity is upper and lower bounded by
\[\frac{(1-\epsilon)\ln (1-\epsilon)}{ c_1\left(
\beta_1^{\frac{2}{\alpha} }+
\beta_2^{\frac{2}{\alpha} } \right)} \frac{B_{TR} + B_{RT}}{F_{Total}}\le C_{\epsilon} \le \frac{(1-\epsilon) \ln (1-\epsilon)}{c_2 \left(
\beta_1^{\frac{2}{\alpha} }+
\beta_2^{\frac{2}{\alpha} } \right)}\frac{B_{TR} + B_{RT}}{F_{Total}} \ \ \text{bits/sec/Hz/$m^2$}, \] 
where $c_1$ and $c_2$ are constants, and $c_2/c_1 = \frac{1}{2} + \frac{1}{\alpha}$.
\end{thm}
\begin{proof}
With $P_{success} = 1-\epsilon$, and using the definition of $C_{\epsilon}$ (\ref{deftc}), the result follows from (\ref{lbound}) and (\ref{upbound}).  
\end{proof}

{\bf Discussion:} In this section we derived an upper and lower bound on the two-way transmission capacity. 	
The upper and lower bound only differ by a constant, and, most importantly, both have identical dependence on the parameters of interest in the two-way communication, $\beta_1$ and $\beta_2$. Thus, the derived bounds establish the two-way transmission capacity up to a constant. The derived upper and lower bounds  for the two-way transmission capacity are
in a fairly simple form and can be used to calculate the two-way
transmission capacity for given rates $B_{TR}$, $B_{RT}$, success
probability $\epsilon$, $F_{TR}$ and $F_{total}$. 
Since the upper and lower bound are identical functions of $\beta_1$ and $\beta_2$, an added advantage of our bounds on the two-way transmission capacity expression 
is that they can be used to find the
optimal value of $F_{TR}$ for given rates $B_{TR}$, $B_{RT}$,
success probability $1=\epsilon$, and $F_{total}$. The optimal
bandwidth allocation that maximizes the two-way transmission
capacity is derived next in the Section \ref{sec:balloc}. 

\section{Two-Way Bandwidth Allocation}
\label{sec:balloc} In Section \ref{sec:twowaytc}, we derived the
two-way transmission capacity of ad-hoc networks within a constant 
as a function of bandwidth allocated to the $Tx_0 \rightarrow Rx_0$
and $Rx_0 \rightarrow Tx_0$ connections. Since the total bandwidth $F_{total}$ is finite, an important question to answer is: what is the optimal
bandwidth allocation between that maximizes the transmission
capacity? For the special case of  equal rate
requirement in both directions, i.e. $B_{TR}=B_{RT}$, equal bandwidth allocation is
optimal. For the non-symmetric case, however, the answer is not that obvious
and is derived in the following theorem.

\begin{thm}\label{thm:SpectrumAllocation}
The optimum bidirectional bandwidth allocation that maximizes the transmission capacity
with two-way communication is
$F_{TR}^\star = x^\star$ and $F_{RT}^\star = F_{RT}^\star$
where $x^\star$ is the unique positive solution to the following equation:
\begin{equation}\label{eq:SpectrumAllocationOptimalSolution}
\frac{1}{B_{TR}}h\left(\frac{B_{TR}}{x}\right)-\frac{1}{B_{RT}}h\left(\frac{
B_{RT}}{F_{total}-x}\right) = 0
\end{equation}
where $h(t) = t^22^t(2^t-1)^{(\delta-1)}$ for $0 < t < F_{total}$.
\end{thm}

\begin{proof}
Neglecting the constant, the two-way transmission capacity  is
\[C = (1-\epsilon) \lambda\left( \frac{B_{TR}+B_{RT}}{F_{total}} \right)=
(1-\epsilon) \frac{- \ln (1-\epsilon)}{ d^2 \left(\left(2^{\frac{B_{TR}}{F_{TR}}}-1\right)^{\frac{2}{\alpha}} + \left(2^{\frac{B_{RT}}{F_{RT}}}-1\right)^{\frac{2}{\alpha}}\right)} \left(\frac{B_{TR}+B_{RT}}{F_{total}}\right).\]
To derive the optimal bandwidth partitioning, i.e. the optimal $F_{TR}$ that
maximizes $C$, we need to minimize $\left(\left(2^{\frac{B_{TR}}{F_{TR}}}-1\right)^{\frac{2}{\alpha}} + \left(2^{\frac{B_{RT}}{F_{RT}}}-1\right)^{\frac{2}{\alpha}}\right)$.

Let $\delta \bydef \frac{2}{\alpha}$. Let $f(x) \bydef \left(\left(2^{\frac{B_{TR}}{x}}-1\right)^{\delta} + \left(2^{\frac{B_{RT}}{F_{total}-x}}-1\right)^{\delta}\right)$. Thus, the problem we need to solve is \[\min_{x\in (0,F_{total})} f(x).\]

The first-order derivative of $f(x)$ is $\frac{\rm d}{{\rm d}x}f(x) = \delta\log_e2\left[-\frac{1}{B_{TR}}h\left(\frac{B_{TR}}{x}\right)+\frac{1}{B_{RT}}h\left(\frac{B_{RT}}{F_{total}-x}\right)\right]$, where $h(t)\bydef t^22^t(2^t-1)^{(\delta-1)}$ for $t \geq 0$. The second-order derivative of $f(x)$ is $\label{eq:f1SecondDerivative}
\frac{{\rm d}^2}{{\rm d}x^2}f(x) = \delta\log_e2\left[\frac{1}{x^2}h\left(\frac{B_{TR}}{x}\right)+\frac{1}{(F_{total} - x)^{2}}h\left(\frac{B_{RT}}{F_{total}-x}\right)\right]$.
Since $h(t)$ is monotonically increasing in $t$ over $t \geq 0$, then we have $h(t) > h(0) = 0$ for all $t >0$. Therefore, $\frac{{\rm d}^2}{{\rm d}x^2}f(x) > 0$ for all $x \in (0,F_{total})$. This means that $f(x)$ is a convex function of $x$ over $(0,F_{total})$ and its minimum corresponds to
$x^\star$ that is the unique positive solution of the following equation $\frac{{\rm d}}{{\rm
d}x}f(x) = 0$, or equivalently, $\frac{1}{B_{TR}}h\left(\frac{B_{TR}}{x}\right)-\frac{1}{B_{RT}}h\left(\frac{B_{RT}}{F_{total}-x}\right) = 0$.
\end{proof}

{\bf Discussion:}
In Theorem \ref{thm:SpectrumAllocation} we derived the optimal
bandwidth allocation for two-way communication in ad-hoc networks that
maximizes the transmission capacity. The result is derived by showing
that the optimization problem is convex in one variable, hence the optimal
solution corresponds to the value for which the function derivative is
zero.
%An important point to note here is that the feasible set for $F_{TR}$ is $[0, F_{total}]$, and we need to guarantee that the optimal solution obtained by Theorem \ref{thm:SpectrumAllocation} lies in  $[0, F_{total}]$.
%For the considered problem the function to be optimized is U shaped
%in the set $[0, F_{total}]$, and hence the value at which the function derivative is zero lies in $[0, F_{total}]$.

Using Theorem \ref{thm:SpectrumAllocation}, if the traffic is symmetric, i.e., $B_{TR} = B_{RT}$, the optimal strategy is naturally
allocate equal bandwidths for two directions with
$F_{TR} = F_{total}/2$.
This result is intuitive since the counterpart parameters in two directions are equal. For asymmetric traffic $B_{TR} \ne B_{RT}$, however, allocating bandwidths proportional to
the desired rate in each direction $F_{TR} = \frac{F_{\rm total}B_{TR}}{B_{TR}+B_{RT}}$ does not satisfy (\ref{eq:SpectrumAllocationOptimalSolution}).
Thus the proportional bandwidth allocation policy is not optimal for
asymmetric traffic for maximizing the transmission capacity, and (\ref{eq:SpectrumAllocationOptimalSolution}) must be satisfied to find the optimal policy.

\section{Effect of Limited Feedback on Two-Way Transmission Capacity with Beamforming}
\label{sec:LimitedFeedback}
In this section we consider an ad-hoc network where each transmitter $Tx_n$ is equipped with $N$ antennas 
while each receiver $Rx_n$ has a single antenna. All other system parameters and assumptions remain the same as defined in Section \ref{sec:sys}.
With multiple transmit antennas, and channel state information CSI at each transmitter, transmission rate can be increased by transmitting the signal along the strongest eigenmode of the channel (called beamforming).  
Beamforming, however, requires that the transmitter  know the channel coefficients, which in general is a challenging problem. 
In a FDD system, the transmitter can learn the channel coefficients, or equivalently the optimum beamformer, through the use of a finite rate feedback channel from the receiver.  Assuming a genie aided feedback (channel coefficients are exactly known at the transmitter, and without accounting for the feedback bandwidth, and SIR required for the feedback),  \cite{Hunter2008} derived the transmission capacity with beamforming, and  showed that the transmission capacity increases as $N^{\frac{2}{\alpha}}$ with increasing $N$. 
In reality, however, feedback requires sufficient bandwidth, and the channel coefficients can be fed back only up to a certain precision.    

Limited feedback techniques \cite{Love2003} are commonly used in practical systems to exploit finite rate feedback channels. 
With limited feedback, a beamforming codebook is assumed known to both the receiver and the transmitter. The receiver computes the best beamforming vector from the beamforming codebook and sends the index of this vector back to the transmitter. The larger the codebook size, the better is the quality of feedback, and consequently better is the data rate from the transmitter to the receiver with beamforming. With a codebook size of $2^B$, each codeword requires $B$ bits of feedback. Thus, the use of a large codebook increases the required bandwidth for the feedback channel, thereby restricting the bandwidth allocated for transmitter to receiver communication. 
Thus, there is a three-fold tradeoff between the bandwidth allocated in forward channel, the feedback channel, and the size of the codebook. In this section, we  quantify this tradeoff and evaluate its effect on the two-way transmission capacity.

The received signal at receiver $R_0$ over bandwidth $F_{TR}$  is 
\begin{eqnarray*}
y_0 = \sqrt{P_t}d^{-\alpha/2}\bh^T_0 \bb_0 x_0 + \sum_{T_n \in \Phi \backslash \{T_0\}}\sqrt{P_t}d_{Tn}^{-\alpha/2}\bh^T_{0n}\bb_n x_n,
\end{eqnarray*}
where $P_t$ is the transmit power of each transmitter, $\bb_n$ are the beamformers used by $Tx_n$, $\bh_0 \in \bbC^{N\times 1}$ is the channel between $Tx_0$ and $Rx_0$, $\bh_{0n}\in \bbC^{N\times 1}$ is the channel between $Tx_n$ and $Rx_0$, 
$d_n$ is the distance between $Tx_n$ and $Rx_0$, $x_0$ and $x_n$ are the data symbols transmitted from $Tx_0$ and $Tx_n$, respectively. For simplicity we assume that each receiver computes the beamforming vectors $\bb_n$ only depending on $\bh_n$, independent of the interferers' channels. 

The received signal at transmitter $Tx_0$ corresponding to 
the feedback by receiver $Rx_0$ over bandwidth $F_{RT}$ is 
\begin{eqnarray}
\label{rxfb}
\bw_0 = \sqrt{P_t}d^{-\alpha/2}\bg_0u_0 + 
\sum_{T_n \in \Phi \backslash \{T_0\}} 
\sqrt{P_t}d_{Rn}^{-\alpha/2} \bg_{0n}u_n,
\end{eqnarray} 
where $\bg_0 \in \bbC^{N\times 1}$ is the channel between $Rx_0$ and $Tx_0$, 
$\bg_{0n}\in \bbC^{N\times 1}$ is the channel between $Rx_n$ and $Tx_0$, 
$u_0$ and $u_n$ are the feedback 
signals transmitted by $Rx_0$ and $Rx_n$, respectively.

With genie-aided feedback, the optimal beamforming vector $\bb_n$ is known to be $\bb_n = \bh_n^*$. In practice, however, only a finite number of bits are available for 
feedback, and hence $\bb_n$ can be modeled as $\bb_n = \bh_n^* + \bee$, where $\bee$ is the additive error term which represents the uncertainty due to limited feedback. The quantization error $\bee$ degrades the signal power compared to genie aided feedback. With $B$ bits of feedback bits, the signal power \cite{Mondal2006}  is 
$|\bh_n|^2\left(1 -c_3\left(\frac{1}{B}\right)^{\frac{1}{N-1}}\right)$ ($c_3<1$ is a constant), compared to $|\bh_n|^2$ for genie aided feedback ($B=\infty$). Thus, the 
SIR for $Tx_0$ to $Rx_0$ communication with $B$ bits of feedback is 
\begin{equation}
\label{sirdata}SIR_{TR} = \frac{d^{-\alpha}|\bh^T_0|^2\left(1 - c_3\left(\frac{1}{B}\right)^{\frac{1}{N-1}}\right)}
{ \sum_{T_n \in \Phi \backslash \{T_0\}}d_{Tn}^{-\alpha}|\bh^T_{0n}\bb_0|^2 },
\end{equation}
and the corresponding  mutual information from $Tx_0$ to $Rx_0$ using bandwidth $F_{TR}$ is 
\[MI_{TR} \bydef F_{TR}\log \left(1+ SIR_{TR}\right) \  \text{bits/sec}.\]

Similarly, the SIR for the feedback link is $SIR_{RT} = \frac{d^{-\alpha}|\bg_0(1)|^2}
{\sum_{T_n \in \Phi \backslash \{T_0\}}d_{Rn}^{-\alpha}|\bg_{0n}(1)|^2 }$, and thus with bandwidth $F_{RT}$, the mutual information of the feedback link is 
\[MI_{RT} \bydef \left(F_{RT}\right)\log \left(1+ SIR_{RT} \right).\]

Similar to (\ref{psucstwoway}), we define the success probability as  the probability that communication in both directions is successful
simultaneously, i.e.
\[P_{success} = P\left(MI_{TR} > B_{TR}, MI_{RT}\ge B\right).\] Consequently, with $P_{success} = (1-\epsilon)$ the two-way transmission capacity is defined as 
\[C_{\epsilon} = \frac{\lambda  (1-\epsilon) B_{TR}}{F_{total}} \ \ \ \text{bits/sec/Hz/m$^2$}.\]

As stated before, in a two-way communication model, where the transmitter locations are modeled as a PPP, the interference received in both directions is correlated. Therefore, computing the success probability in closed form is a hard problem. To derive a meaningful insight into the dependence of bandwidth allocation, and  feedback bits on two-way transmission capacity, we derive a lower bound on the success probability using the FKG inequality as follows. 
\begin{thm} Accounting for feedback bandwidth, the two-way transmission capacity with beamforming is lower bounded by 
\[C_{\epsilon}\ge \frac{(1-\epsilon)\epsilon N^{\frac{2}{\alpha}}}{c_4[(\beta_1/\gamma)^{\frac{1}{\alpha}} + (\beta_3)^{\frac{1}{\alpha}}]}\frac{B_{TR}}{F_{total}} \ \ \text{bits/sec/Hz/m$^2$},\]
where $\gamma \bydef \left(1- c_3\left(\frac{1}{B}\right)^{\frac{1}{N-1}}\right)$, and 
\[c_4= \left( \left(1+\sum_{k=0}^{N-2} \frac{1}{(k+1)!}\prod_{\ell=0}^k\left(\ell-\frac{2}{\alpha}\right)\right)  \left(           \frac{2\pi}{\alpha}\sum_{k=0}^{N-1}{N\choose k} B\left(\frac{2}{\alpha}+k; N-\frac{2}{\alpha} +k \right)  \right) \right)^{-1}\] with $B(a,b) = \frac{\Gamma(a)\Gamma(b)}{\Gamma(a+b)}$.

\end{thm}
\begin{proof}
\[P_{success} = P\left(MI_{TR} > B_{TR}, MI_{RT}\ge B\right).\] 
Similar to Example \ref{exmdecevent}, the success events in two directions $\{MI_{TR} > B_{TR}\}$, and $\{MI_{RT} > B\}$, respectively, are decreasing events. Thus, invoking Lemma \ref{lemfkg},
\[P_{success} \ge P\left(MI_{TR} > B_{TR}\right)P\left( MI_{RT}\ge B\right).\]
By definition, 
\begin{eqnarray*}
P\left(MI_{TR} > B_{TR}\right) &=& P\left(F_{TR}\log \left(1+ SIR_{TR}
\right)> B_{TR}\right), \\
&\stackrel{(a)}{=}& P\left(SIR_{TR} >\beta_1\right),\\
&\stackrel{(b)}{=}& P\left(\frac{d_0^{-\alpha}|\bh^T_0|^2\left(1 - |\bh_n|^2\left(\frac{1}{B}\right)^{\frac{1}{N-1}}\right)}
{ \sum_{T_n \in \Phi \backslash \{T_0\}}d_n^{-\alpha}|\bh^T_{0n}\bb_0|^2 }>\beta_1\right),\\
&\stackrel{(c)}{\ge}&P\left(\frac{d_0^{-\alpha}|\bh^T_0|^2 }{ \sum_{T_n \in \Phi \backslash \{T_0\}}d_n^{-\alpha}|\bh^T_{0n}\bb_0|^2 }>\beta_1/\gamma\right),\\
&\stackrel{(d)}{=}& 1-c_4\lambda(\beta_1/\gamma)^{\frac{2}{\alpha}}N^{\frac{-2}{\alpha}},
\end{eqnarray*}
where $(a)$ follows from the definition of $\beta_1$, $(b)$ follows by substituting for $SIR_{TR}$ (\ref{sirdata}), $(c)$ follows by defining $\gamma \bydef \left(1- c_3\left(\frac{1}{B}\right)^{\frac{1}{N-1}}\right)$, and $(d)$ follows from Theorem $3$ \cite{Hunter2008}.

Directly applying Theorem $3$ \cite{Hunter2008}, 
$P\left( MI_{RT}\ge B\right) = 1- c\lambda(\beta_3)^{\frac{2}{\alpha}}N^{\frac{-2}{\alpha}}$, 
where $\beta_3 = d^{\alpha}\left(2^{\frac{B}{F_{RT}}}\right)$

Thus, 
$P_{success} \ge 1- c_4\lambda N^{\frac{-2}{\alpha}}\left[ (\beta_1/\gamma)^{\frac{2}{\alpha}} +(\beta_3)^{\frac{2}{\alpha}}\right]$.
Then,
\[C_{\epsilon}\ge \frac{(1-\epsilon)\epsilon N^{\frac{2}{\alpha}}}{c_4[(\beta_1/\gamma)^{\frac{1}{\alpha}} + (\beta_3)^{\frac{1}{\alpha}}]}\frac{B_{TR}}{F_{total}} \ \ \text{bits/sec/Hz/m$^2$}.\]
\end{proof}
{\it Discussion:} In this section we derived a lower bound on the two-way transmission capacity when the transmitter uses beamforming with limited feedback, as a function of the bandwidth allocated in two directions, and the number of feedback bits. 
Note that as $B$ (the number of feedback bits) increases, the two-way transmission capacity increases as $B^{\frac{1}{(N-1)\alpha}}$ due to the improvement in signal strength, however, decreases as $2^{\frac{-B}{\alpha}}$ because of the stringent requirement of SIR on the feedback link to be more than $\beta_3$. Our result quantifies the degradation due to practical limited feedback in two-way transmission capacity with beamforming, compared to assuming a genie aided feedback \cite{Hunter2008}. The feedback requirement not only decreases the available bandwidth for transmitter to receiver communication, but also degrades the overall  performance due to the successful reception requirement of the feedback bits. 

Similar to Section \ref{sec:balloc}, for a fixed value of $B$ and $B_{TR}$, the optimal bandwidth allocation $F_{TR}$ that maximizes the two-way transmission capacity upper bound can be computed using Theorem \ref{thm:SpectrumAllocation}, since here again the optimization problem is convex. For a fixed value of $F_{TR}$ and $B_{TR}$, finding the optimal $B$ is slightly complicated since the upper bound is not a convex function of $B$, however, the problem is a single variable problem and can be solved easily by using techniques like bisection.

\section{Numerical  Results}
\label{sec::NumericalSimulationResults} In this section we present some numerical results on the two-way transmission capacity. 
We adopt the simulation methodology for one-way networks
presented in \cite{WeberKam05} and consider $d = 5$m, and  $\alpha = 4$.
\subsection{General Two-way Communication}\label{subsec:simGeneralTwoway}
\textbf{Tightness of the proposed bounds: }
In this experiment, we consider $B_{\rm TR} = 1.028$ kbits, $B_{\rm RT} = 0.03$ kbits, $F_{\rm TR} = 0.99$ MHz, and $F_{\rm RT} = F_{\rm total} - F_{\rm TR} = 0.01$ MHz. Fig. \ref{fig:GenBoundTightness} shows the curves for the simulated result and the bounds derived in Theorem \ref{thm:bounds} on the two-way transmission capacity as functions of the outage probability requirement.  Moreover, note that the transmission capacity decreases at very high outage probability ($\epsilon$), since the transmission capacity expressions are proportional to $-(1-\epsilon)\log(1-\epsilon)$. Intuitively, as the outage probability $\epsilon$ approaches towards 1, a high density of links is allowed in a unit area, however, most of the links fail; therefore, the amount of successfully received information actually decreases.

\textbf{One-way versus two-way transmission capacity:} Requiring that transmissions be successful in both directions, the two-way transmission capacity is less than the one-way transmission capacity. To quantify the loss we plot the two-way transmission capacity in comparison with the one-way transmission capacity for the same total bandwidth $F_{\rm total}$ and total data rates $B_{\rm total} = B_{\rm TR}+B_{\rm RT}$. In particular, for the results shown in Fig. \ref{fig:GenOnevsTwo}, we set $B_{\rm TR}=1.024$ kbits, $B_{\rm RT} = 0.256$ kbits, $F_{\rm TR} = 0.8$ MHz, and $F_{RT} = 0.2$ MHz. The simulation results show that at the outage requirement of $10\%$, the two-way transmission capacity is half the one-way transmission capacity.

\textbf{Effect of bandwidth allocation:} To highlight the effect of bandwidth allocation on the two-way transmission capacity we plot the transmission capacity as a function of $F_{\rm TR}$ in Fig. \ref{fig:GenBWAllocation} assuming the total bandwidth is $F_{\rm total} = 1$ MHz. For the scenario of symmetric traffic, we set the data requirements in two directions equal to 1 kbits, i.e., $B_{\rm TR} = B_{\rm RT} = 1.024$ kbits. In this case, we notice that the proportional allocation method is optimal. For asymmetric traffic, we consider $B_{\rm TR} = 1.024$ Mbits and $B_{\rm RT} = 0.056$ kbits. From Fig. \ref{fig:GenBWAllocation}, note that the optimal bandwidth allocation (Theorem \ref{thm:SpectrumAllocation}) provides a gain of $36\%$ over the proportional allocation.

\subsection{Feedback-Based Communication}\label{subsec:simFeedback}
To quantify the effect of feedback on the transmission capacity we compare the transmission capacity of a feedback-based network with the corresponding one-way network with the genie-aided beamforming \cite{Hunter2008} with  $N=3$ in Fig. \ref{fig:FBOnevsTwo}. 
We use $B_{\rm TR} = 1.024$ kbits, $B_{\rm RT} = 0.056$ kbits, feedback bits $B=2$, $F_{\rm TR} = 0.94$ MHz, and $F_{\rm RT} = F_{\rm total} - F_{\rm TR} = 0.06$ MHz and assume that the transmitters employ Grassmannian limited feedback beamforming for transmission \cite{Love2003}. Moreover, of the $B_{\rm RT}=0.056$ kbits (or 56 bits) in the reverse direction, $B$ bits are used for carrying the codeword index while the other bits are used for MAC header. 

\textbf{Tightness of the proposed lower bound: } In this experiment, we set $N = 3$ antennas and $B = 2$ bits. Fig. \ref{fig:FBBoundTightness} presents the simulated results for a genie-aided beamforming network and the limited-feedback beamforming network as well as the computed lower bound.
%\textbf{Effects of the beamforming codebook size: } For a given number of antennas, using a large codebook size provides the transmitters with a better beamforming vector in the sense of closer to the unquantized beamforming vector. Fig. \ref{fig:FBVaryingCodebookSize} shows the simulated results for the networks with $N = 3$. Note that the Grassmannian beamforming codebook with $B=4$ has a performance close to the genie-aided beamforming case.
%\textbf{Effects of the number of antennas: } We investigate the transmission capacity of feedback-based networks with various numbers of antennas assuming that $B = 4$ bits. From Fig. \ref{fig:FBVaryingAntennaNumber}, we notice that the two-way transmission capacity increases with the number of antennas.

\vspace{-.2in}
\section{Conclusions}\label{sec:Conclusions}
In this paper we generalized the concept of transmission capacity to incorporate two-way
communication in wireless ad-hoc networks. The two-way transmission capacity is able to
capture the requirement of successful transmissions in both directions and the impact of duplexing
techniques. The two-way success requirement is shown to reduce the transmission capacity significantly compared to the corresponding one-way transmission capacity.
This observation raised the question of finding the network with the maximum
two-way transmission capacity among the two-way networks with the same total bandwidth given fixed desired rates in two directions. We addressed the question by providing the optimal solution for bidirectional spectrum allocation  to maximize the two-way transmission capacity. The optimal solutions were determined in terms of the
path-loss exponent, desired rates, and total bandwidth available.

As an application of the two-way transmission capacity framework, we also quantified the effect of practical limited channel feedback on the two-way transmission capacity with transmit beamforming. We showed that accounting for the bandwidth required for feedback, and the successful reception of the feedback bits, the transmission capacity is significantly reduced compared to the genie aided feedback.

\bibliographystyle{../IEEEtran}
\bibliography{../IEEEabrv,../Research}

\begin{figure}
\centering
\includegraphics[width=3.5in]{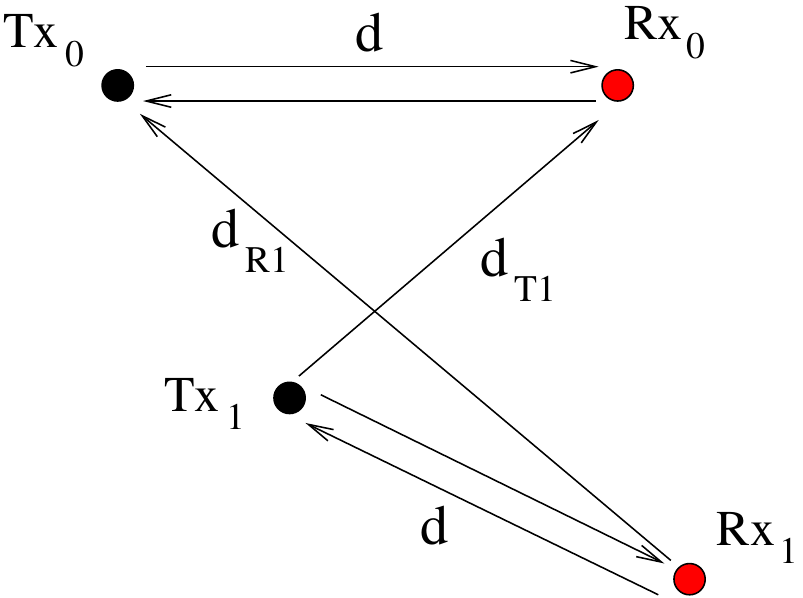}
\caption{Schematic of two-way communication with two pairs of nodes.}
\label{intfindp}
\end{figure}
\begin{figure}[!h]
\begin{centering}
\includegraphics[width=3.2in]{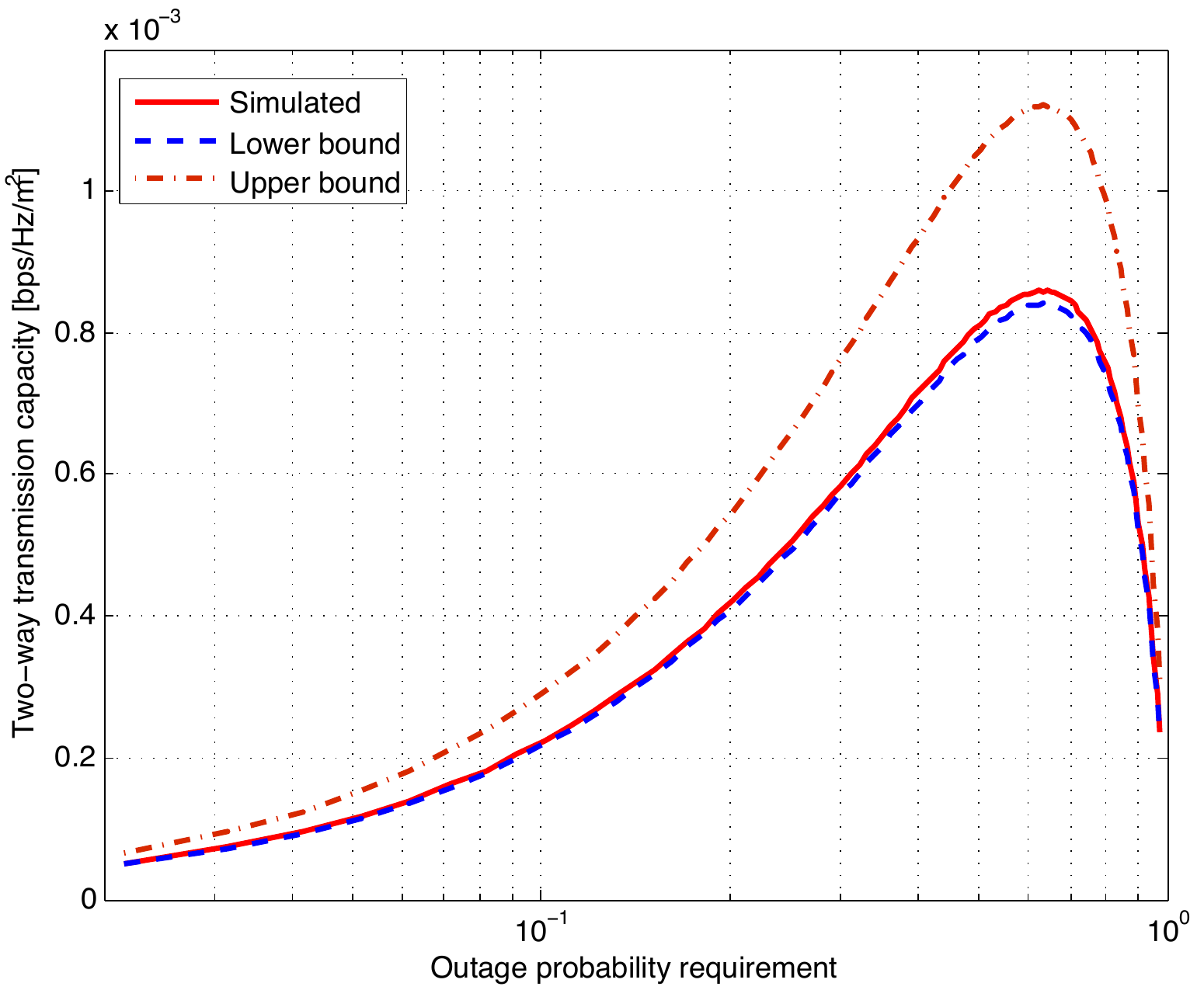}
\caption{Tightness of the proposed bounds on the transmission capacity of general two-way networks.}
\label{fig:GenBoundTightness}
\end{centering}
\end{figure}
\begin{figure}[!h]
\begin{centering}
\includegraphics[width=3.2in]{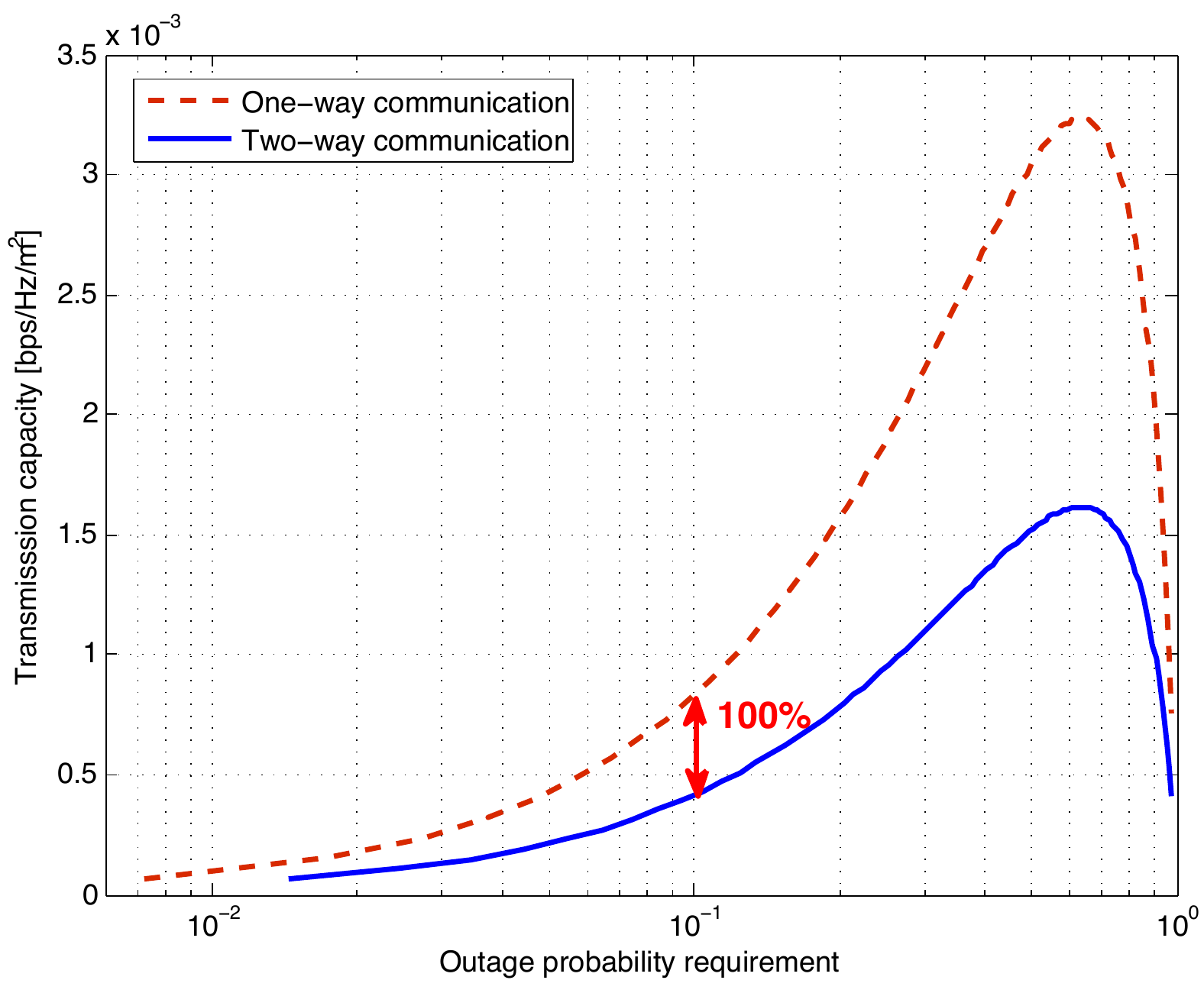}
\caption{Comparison of the one-way transmission capacity and the general two-way transmission capacity.}
\label{fig:GenOnevsTwo}
\end{centering}
\end{figure}

\begin{figure}[!h]
\begin{centering}
\includegraphics[width=3.2in]{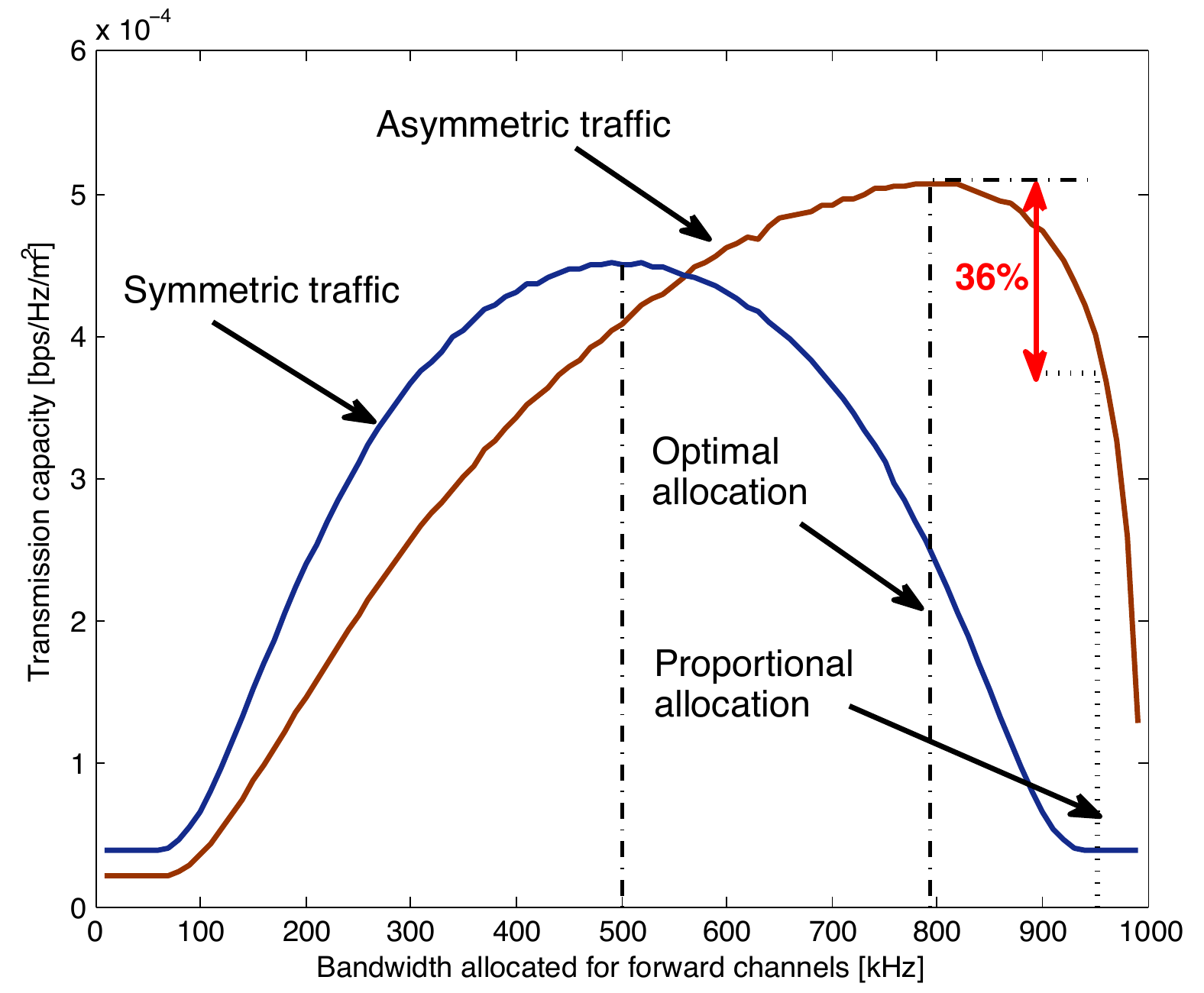}
\caption{Two-way transmission capacity as a function of bandwidth allocation. For symmetric traffic, the proportional allocation method is optimal, while for asymmetric traffic the optimal allocation provides a large gain over the proportional allocation.}
\label{fig:GenBWAllocation}
\end{centering}
\end{figure}

\begin{figure}[!h]
\begin{centering}
\includegraphics[width=3.2in]{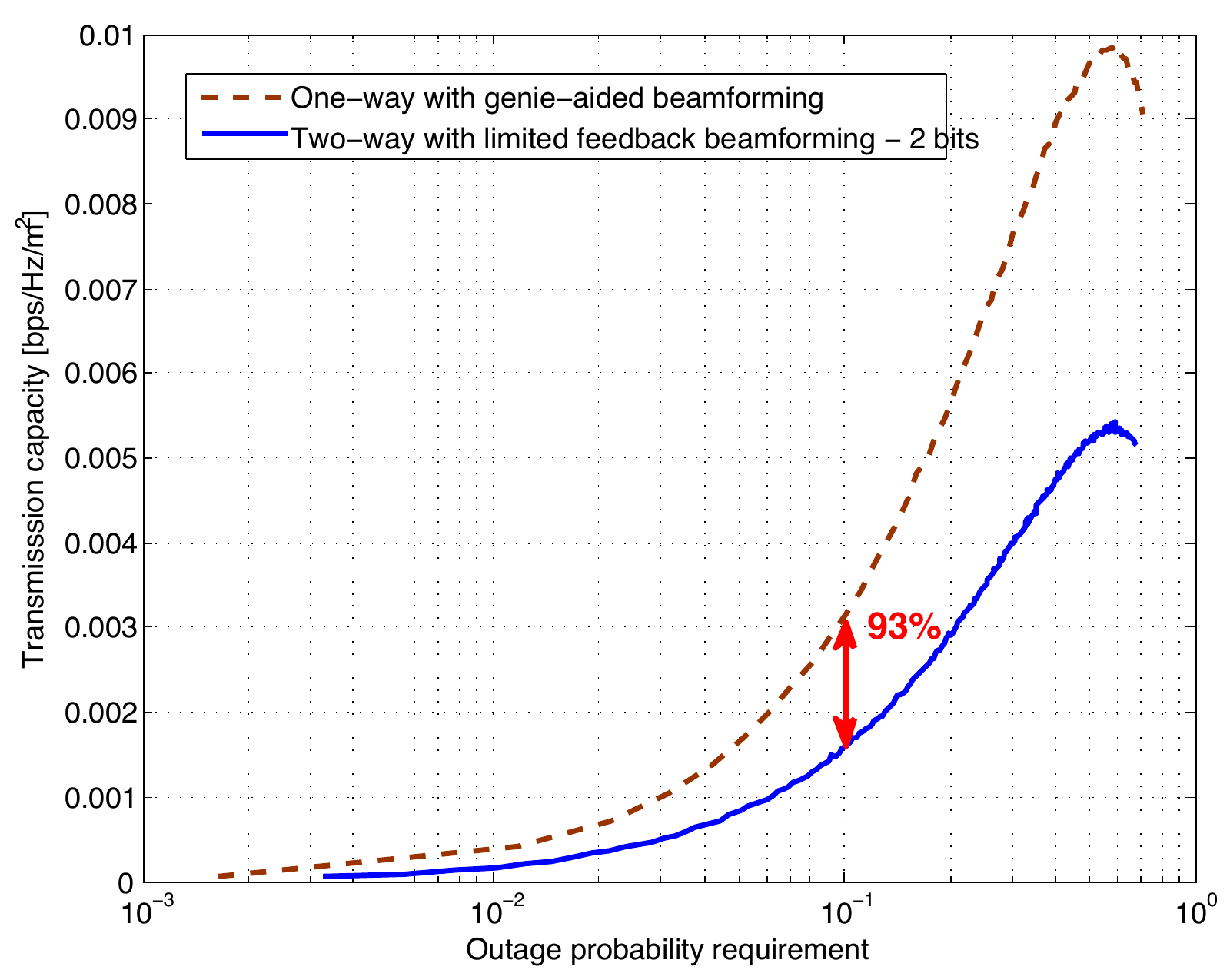}
\caption{Comparison of the transmission capacity of a feedback-based network with that of the corresponding one-way network.}
\label{fig:FBOnevsTwo}
\end{centering}
\end{figure}

\begin{figure}[!h]
\begin{centering}
\includegraphics[width=3.2in]{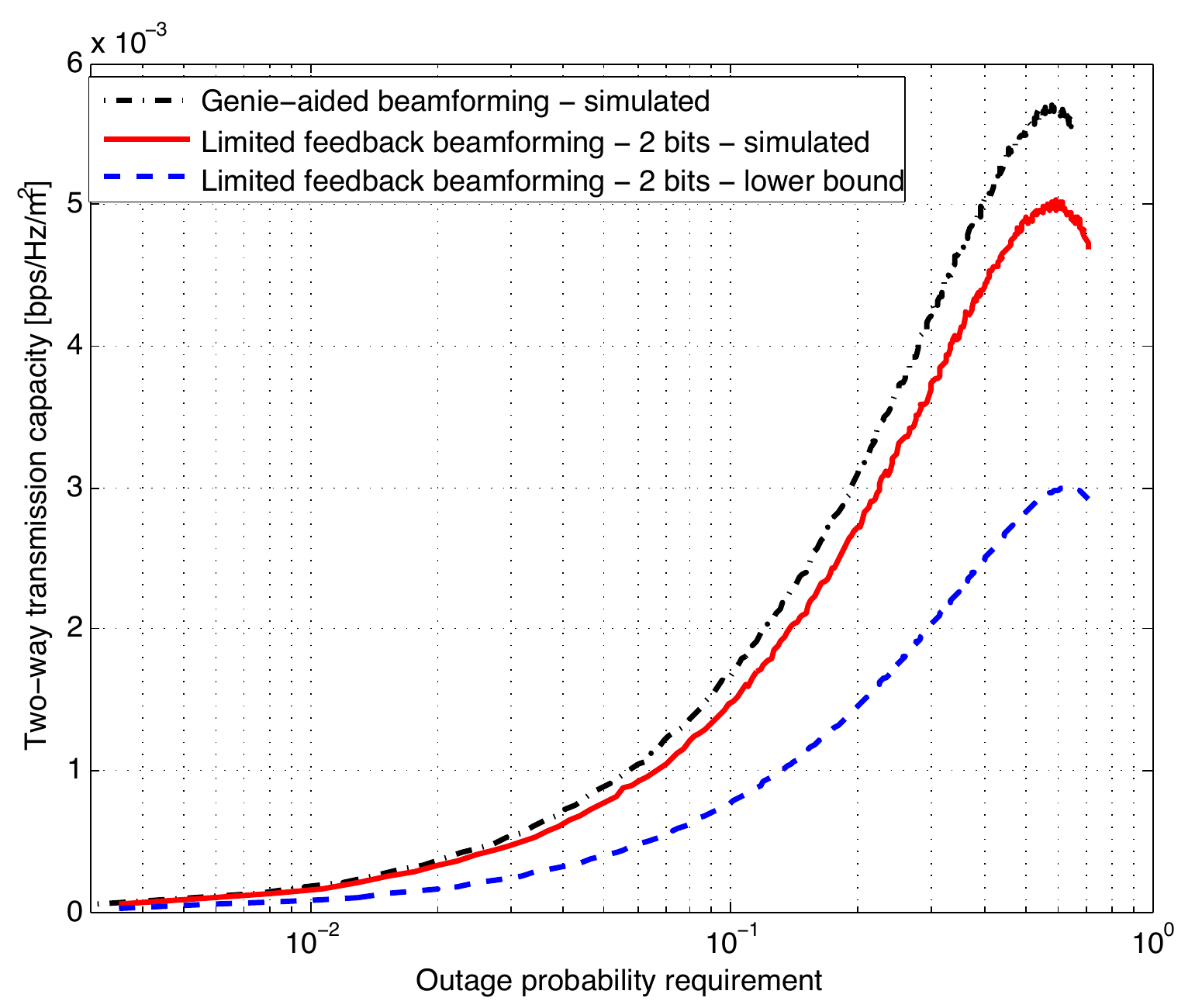}
\caption{Tightness of the proposed lower bound on the transmission capacity of feedback-based networks.}
\label{fig:FBBoundTightness}
\end{centering}
\end{figure}

\end{document}